%% file: arxiv.tex
\documentclass[a4paper]{article}

\usepackage[margin=1in]{geometry}
\usepackage[ruled,linesnumbered,noend]{algorithm2e}
\usepackage{amsthm}
\usepackage{amsmath}
\usepackage{amssymb}
\usepackage{todonotes}

\newcommand{\RR}{\mathbb{R}}
\newcommand{\NN}{\mathbb{N}}
\newcommand{\OPT}{\operatorname{OPT}}
\newcommand{\cost}{\operatorname{cost}}

\newtheorem{theorem}{Theorem}
\newtheorem{lemma}[theorem]{Lemma}

\newtheorem{definition}[theorem]{Definition}

\newtheorem{claim}[theorem]{Claim}

\newtheorem{result}{Result}

\begin{document}

\title{A Subquadratic Time Approximation Algorithm for Individually Fair k-Center}

\author{
Matthijs Ebbens \thanks{Department of Mathematics and Computer Science, University of Cologne, Germany. Email: \texttt{ymebbens@gmail.com}.}
\and
Nicole Funk \thanks{Department of Mathematics and Computer Science, University of Cologne, Germany. Email: \texttt{funk@cs.uni-koeln.de}.}
\and 
Jan Höckendorff\thanks{Department of Mathematics and Computer Science, University of Cologne, Germany. Email: \texttt{hoeckendorff@cs.uni-koeln.de}. Funded by the Deutsche Forschungsgemeinschaft (DFG, German Research Foundation) – Project Number 459420781.}
\and Christian Sohler \thanks{Department of Mathematics and Computer Science, University of Cologne, Germany. Email: \texttt{sohler@cs.uni-koeln.de}.}
\and Vera Weil \thanks{Department of Mathematics and Computer Science, University of Cologne, Germany. Email: \texttt{weil@cs.uni-koeln.de}.}
}

\date{}

\maketitle

\input{abstract}

\input{introduction}

\input{contributions}

\input{relatedwork}

\input{preliminaries}

\input{algorithm}

\input{subquadratic_algorithm}

\bibliographystyle{elsarticle-num} 
\bibliography{biblio}

\end{document}

%% file: abstract.tex
We study the $k$-center problem in the context of individual fairness.
Let $P$ be a set of $n$  points in a metric space and $r_x$ be the distance between $x \in P$ and its $\lceil n/k \rceil$-th nearest neighbor. The problem asks to optimize the $k$-center objective under the constraint that, for every point $x$, there is a center within distance $r_x$. 
We give bicriteria $(\beta,\gamma)$-approximation algorithms that compute clusterings such that every point $x \in P$ has a center within distance $\beta r_x$ and the clustering cost is at most $\gamma$ times the optimal cost.
Our main contributions are a deterministic $O(n^2+ kn \log n)$ 
time $(2,2)$-approximation algorithm and a randomized $O(nk\log(n/\delta)+k^2/\varepsilon)$ time $(10,2+\varepsilon)$-approximation algorithm, where $\delta$ denotes the failure probability. For the latter, we develop a randomized sampling procedure to compute constant factor approximations for the values $r_x$ for all $x\in P$ in subquadratic time; we believe this procedure to be of independent interest within the context of individual fairness.

%% file: introduction.tex
\section{INTRODUCTION}\label{sec: introduction}

Clustering is the process of partitioning a given data set into subsets such that, ideally, the elements in each subset are similar to each other and elements in different subsets are less similar. It is a basic approach in unsupervised learning with applications in, for example,
signal processing \cite{lloyd1982}, community detection in networks \cite{zhang2007}\cite{bota2015}, outlier detection \cite{chawla2013} and  data summarization \cite{kleindessner2019fair}. 

For many clustering problems each cluster has a corresponding center that can be viewed as a representative of all points contained in the given cluster. Data points closer to the center are naturally better represented than those that are far away. Ideally, we would like every data point to be close to its representative center, but this is usually impossible to realize. One way to mitigate this is to use clustering objectives such as $k$-median or $k$-means clustering that optimize the average (squared) distance to the cluster center. But optimizing the average distance will often lead to some data points being less well represented than they could be. In applications where data points model humans, such a behaviour may be undesirable and raises the question how to obtain a clustering that is acceptable for everyone. This problem is addressed in a recent paper by Jung et al. \cite{jung2020service} that introduced the notion of individual fairness in the context of clustering. 

Individual fairness can be thought of as a clustering constraint that
ensures that every data point is assigned to a center that is not much worse than what one could hope for in a certain average-type setting:
If one partitions a data set with $n$ points in $k$ clusters then 
on average a cluster has $n/k$ data points \cite{jung2020service}. Thus,  individual fairness is
defined to be the constraint that every data point has a center among its $\lceil n/k \rceil$-th nearest neighbors.

Individually fair clustering was mainly studied for $k$-median and $k$-means clustering in metric spaces and generalizations to $l_p$-norm type distance functions
\cite{mahabadi2020individual} \cite{negahbani2021better}
\cite{vakilian2022improved}
\cite{bateni2024scalable}. 
This also led to results on the individually fair $k$-center problem, which we will call the $\alpha$-fair $k$-center problem.

In this paper, we show that the properties of the $\alpha$-fair $k$-center problem can be exploited to achieve a simple algorithm that improves upon the current state of the art. 
We then give another algorithm that achieves a randomized constant bicriteria approximation in $O(kn \log (n/\delta))$ time where $\delta>0$ is a parameter that bounds the failure probability. To do so, we develop a procedure to approximate the distances to the $\lceil n/k \rceil$-th nearest neighbors of each data point in that running time.
We believe that this procedure may find further applications in the context of individual fairness.

%% file: contributions.tex
\section{OUR CONTRIBUTIONS}
In this work, we give a $O(n^2 + kn \log n)$ time $(2,2)$-bicriteria approximation algorithm for the individually $\alpha$-fair metric $k$-center problem, improving upon the current state of the art  $(3,3)$-approximation algorithm by Vakilian and Yal\c{c}iner \cite{vakilian2022improved}.
\begin{result} [Simplified restatement of Theorem \ref{thm: result discrete}]
    Algorithm $(2,2)$-ApproxFairCenter is a $(2,2)$-bicriteria approximation algorithm for the individually $\alpha$-fair (discrete) metric $k$-center problem 
    with
    $O(n^2 + kn \log n)$ running time, where $n$ is the number of input points.
\end{result}

As the running time of this algorithm is dominated by the computation of the $\lceil n/k \rceil$-th nearest neighbors of each data point, we develop a randomized approximation algorithm for these distances under the assumption that $k$ is smaller than $n/6$ (if this is not the case we can compute all distances to the $\lceil n/k\rceil$-th nearest neighbors exactly in $O(nk)$ time).

\begin{result}[Simplified restatement of Lemma \ref{lem: approx Fair Radii}]
Let $(P,d)$ be a finite metric space.
 Algorithm ApproxFairRadii computes with probability at least $1-\delta$ for all
$p\in P$ a $5$-approximation to the distance of the $\lceil n/k\rceil$-th nearest neighbor. The running time of the algorithm is $O(nk \log(n/\delta))$. 
\end{result}

Using the approximated distances, we modify the $(2,2)$-approximation algorithm to obtain a randomized approximation algorithm with subquadratic running time if $k$ is asymptotically smaller than $n$.

\begin{result}[Simplified restatement of Theorem \ref{thm: subquadratic time}]
Let $\delta, \varepsilon>0$.
Algorithm $(10,2+\epsilon)$-ApproxFairCenter is a randomized 
algorithm that computes with probability at least $1-\delta$ a $(10,2+\epsilon)$-bicriteria approximation for the
individually $\alpha$-fair (discrete) metric $k$-center problem.
The running time of the algorithm is
$O(nk (\log(n/\delta)) +k^2/\varepsilon)$, where $n$ is the number of input points.
\end{result}

%% file: relatedwork.tex
\section{RELATED WORK}\label{sec: related work}

For the $k$-center problem there is a $O(nk)$ time $2$-approximation by Gonzalez ~\cite{gonzalez1985clustering} and it is known that the problem is NP-complete to approximate with a factor less than $2$~\cite{hochbaum1985best}.

The notion of individual fairness we consider in this paper was introduced by Jung et al. \cite{jung2020service} for $k$-clustering in general, but also studied previously in other contexts such as priority clustering \cite{plesnik1987heuristic} and metric embedding \cite{chan2006spanners}\cite{charikar2010local}. In fact, individually fair clustering is a special case of priority clustering \cite{bajpai2021revisiting}. In the context of individual fairness, a clustering of size $k$  is \emph{$\alpha$-fair} for a set of points $P$ and some $\alpha > 0$ if for each point $p$ in $P$ there is a center within $\alpha$ times the distance to its $ \lceil|P|/k\rceil$-th nearest neighbor in $P$. Jung et al. \cite{jung2020service} showed that 
solving $k$-clustering under individual fairness with $\alpha = 1$ is NP-hard in any metric space. They additionally designed an algorithm that computes a $k$-center clustering that is $2$-fair but lacks any guarantee on the cost.
In the context of individually fair $k$-clustering, previous work studied algorithms for arbitrary $\ell_p$ norm cost functions including $k$-center, $k$-median and $k$-means as special cases \cite{mahabadi2020individual} \cite{negahbani2021better} \cite{vakilian2022improved}
\cite{bateni2024scalable}. We start by stating the results for the case of $k$-center. The first result with guarantees on both cost and fairness was by Mahabadi and Vakilian \cite{mahabadi2020individual}, who developed a $\tilde{O}(k^5n^5)$-time local search $(O(\log n),7)$-approximation algorithm that computes a $7\alpha$-fair solution with cost at most $O(\log n)$ times the optimal cost, for the case $\alpha \geq 1$.
Further work was made by Negahbani and  Chakrabarty \cite{negahbani2021better}  utilizing LP rounding to develop a $\tilde{O}(kn^4)$-time $(2+\varepsilon,8)$-approximation for $\alpha = 1$ and $\varepsilon > 0$. The most recent result with respect to bicriteria approximation is by Vakilian and Yal\c{c}iner \cite{vakilian2022improved}, who proposed a $\tilde{O}(n^2)$-time $(3+\varepsilon,3)$-approximation for any $\alpha \geq 1$ and $\varepsilon > 0$. Han et al. \cite{han2023approximation} consider a relaxed version of individual fair $k$-center by allowing outliers and propose a $4$-approximation for this variant. 
Next we discuss the results for other $\ell_p$-norms.
The algorithm of Mahabadi and Vakilian \cite{mahabadi2020individual} achieves a $(84,7)$-approximation algorithm for $\alpha$-fair $k$-median. Negahbani and Chakrabarty \cite{negahbani2021better} designed a $(2^{p+2},8)$-approximation, which yields cost approximations $8$ for $k$-median ($p = 1)$ and $16$ for $k$-means $(p = 2)$ objective functions. 
Further results by Vakilian and Yal\c{c}iner \cite{vakilian2022improved} 
are a polynomial time $(16^p+\varepsilon,3)$-approximation for general $p$ and a $(7.081 + \varepsilon,3)$-approximation for $k$-median. Finally, for $\alpha = 1$, Bateni et al. \cite{bateni2024scalable} proposed an $\tilde{O}(nk^2)$ time $(O(1),6)$-approximation algorithm for $k$-median and $k$-means. 
Additional studies include coresets for individually fair $k$-clustering by Chhaya et al. \cite{chhaya2022coresets} and alternative notions to individual fairness like a feature based approach by Kar et al. \cite{kar2021feature} and an approach through similarity sets by Chakrabarti et al. \cite{chakrabarti2022new}. 

%% file: preliminaries.tex
\section{PRELIMINARIES}

In this section, we will establish the necessary notation and introduce the $k$-center problem with individual fairness.
Throughout this work $(P,d)$ will denote an arbitrary discrete metric space with $n = |P|$. $P$ will be given as input and we assume that one can compute $d(x,y)$ in constant time for any $x,y \in P$. For a set $C \subseteq P$ and some point $p \in P$ we write $d(p,C) := \min_{c \in C}d(p,c)$ and set $d(p,\emptyset):=\infty$.
We denote the cost of a set $C \subseteq P $ by $
\text{cost}(P,C) := \max_{p \in P} d(p,C)$.
The definition of the unconstrained $k$-center problem is then as follows.
\begin{definition}[discrete metric $k$-center problem]
Given a discrete metric space $(P,d)$ and a target number of clusters $k \in \NN$, find a set $C \subseteq P$ of size $k$ such that $\text{cost}(P,C)$ is minimized.
\end{definition}

The individual fairness constraint is modelled in the following way
\cite{jung2020service}.

\begin{definition}(fairness radius)
Let $k\in \NN$, $p \in P$ and let $B(p,r) := \{q \in P \mid \ d(p,q) \leq r\}$ be the ball centered at $p$ with radius $r \in \RR_{\geq 0}$. The \emph{$k$-fair radius} of $p$ is defined as 
\[
r_{k}(p) := \inf \{r \in \RR_{\geq 0}  \mid \ \vert B(p,r) \vert \geq n/k\}.
\]
\end{definition}

\begin{definition} (individually fair centers)
For $\alpha \in \RR_{> 0}$ and $k \in \NN$, a set $C \subseteq P$ of centers is called  \emph{$(\alpha,k)$-fair}, if 
$d(p,C) \leq \alpha  r_k(p)$ for all $p \in P$.
If in addition $|C|\le k$, then $C$ is called a
feasible solution for the  $\alpha$-fair $k$-center problem.
\end{definition}

When the parameters $\alpha$ and $k$ are clear from the context we call a feasible solution to the $\alpha$-fair $k$-center problem also simply a \emph{feasible solution}.

\begin{definition}[\emph{$\alpha$-fair $k$-center problem}]
Given a discrete metric space $(P,d)$, a target number of clusters $k \in \NN$, 
and a fairness parameter $\alpha \in \RR_{>0}$,
find an $(\alpha,k)$-fair set $C \subseteq P$ of size $k$ that minimizes
$
\text{cost}(P,C).
$
\end{definition}
Note that there does not necessarily exist a feasible solution for all choices of $\alpha>0$ in the $\alpha$-fair $k$-center problem \cite{jung2020service}. In this case, we define the cost of an optimal solution as infinity. However, it is known that for $\alpha\ge 2$ there is always a feasible solution \cite{jung2020service}. We will consider the following bi-criteria approximation.

\begin{definition}[\emph{($\beta$, $\gamma$)-approximation for $\alpha$-fair $k$-center}]
 Given $k \in \NN$, $\alpha \in \RR_{>0}$
 and let $\Delta^*$ be the cost of an optimal solution for the $\alpha$-fair $k$-center problem.
 A set $C \subseteq P$ of size $k$ is a \mbox{($\beta$, $\gamma$)}-approximation for the $\alpha$-fair $k$-center problem, if the following two conditions hold:
 \begin{enumerate}
     \item $cost(P,C) \leq \beta \Delta^*$,
     \item $C$ is $(\gamma \alpha,k)$-fair.
 \end{enumerate}
\end{definition}

Note that, since $\alpha$ is part of the input, one can reduce the $k$-center problem to the $\alpha$-fair $k$-center problem, by choosing a suitable $\alpha$ for given input points s.t. any set of $k$ centers is $(\alpha,k)$-fair. By the hardness results on individual fairness \cite{jung2020service} and approximability of $k$-center \cite{hochbaum1985best}, a $(2,2)$-approximation algorithm for $\alpha$-fair $k$-center is the best one can hope for in polynomial time unless P equals NP.

%% file: algorithm.tex
\section{$(2,2)$-APPROXIMATION ALGORITHM}\label{sec: algorithm}

In this section, we present and analyze our deterministic bicriteria approximation algorithm for the $\alpha$-fair $k$-center problem. 
The first algorithm for individually fair clustering by Jung et al. \cite{jung2020service} computes a $(2,k)$-fair solution by considering points $p \in P$ in order, non-decreasingly sorted by their fairness radii $r_k(p)$, and choosing $p$ as center, if 
$B(p, r_k(p))$ does not intersect the ball $B(q, r_k(q))$, i.e 
$d(p,q) > r_k(p) + r_k(q)$, for all previously chosen centers $q \in P$.

This approach is equivalent to the algorithm for finding
a $(k/n)$-density net by Chan et al. \cite{chan2006spanners}, the metric embedding algorithm of Charikar et al. \cite{charikar2010local} and the $2$-approximation algorithm for priority $k$-center by Plesník \cite{plesnik1987heuristic}\cite{bajpai2021revisiting}.
Vakilian and Yal\c{c}iner  \cite{vakilian2022improved} later changed the condition of choosing $p$ as a center to $d(p,q)$ having to be at least $2 \alpha r_k(p) \geq 2\alpha (r_k(p) + r_k(q))$ for all previously chosen centers $q \in P$.

We adapt this approach by guessing an upper bound $\Delta$ on the cost of an optimal solution and modify the condition to $d(p,q) \geq 2\min \{ \alpha r_k(p),  \Delta \}$, as done in Algorithm $\ref{alg:FairCenter}$.

Besides the fairness parameter $\alpha > 0$ and the number of centers $k$, Algorithm \ref{alg:FairCenter} also receives a cost value $\Delta \in \RR$ and function $r: P\rightarrow \RR_{\geq 0}$ as input. The algorithm is guaranteed to always compute a $(2\alpha,k)$-fair set $C\subseteq P$ with cost at most $2\Delta$ if $r = r_k$. If there exists a feasible solution for $\alpha$-fair $k$-center with cost at most $\Delta$, then it
additionally holds that $|C| \leq k$. 
Assume that we know the optimal value $\Delta^*$ for the $\alpha$-fair $k$-center problem, then we can apply Algorithm \ref{alg:FairCenter} to obtain a $(2,2)$-approximation. To find the value $\Delta^*$ we exploit the fact that there always exist two input points $p,q\in P$
with $d(p,q)= \Delta^*$.

\IncMargin{1.5em}
\begin{algorithm}
\DontPrintSemicolon
\caption{FairCenter\label{alg:FairCenter}}
\KwIn{Number of centers $k \in \NN$, fairness parameter $\alpha \in \RR_{> 0}$, cost value $\Delta \in \RR$, function $r: P \rightarrow \RR_{\geq 0}$}
\KwOut{Set $C\subseteq P$}
Let $P =\{p_1,\dots,p_n \}$ be sorted s.t. $r(p_1) \leq r(p_2) \leq \dots \leq r(p_n)$
    
$C := \emptyset$ \;

\For{$i := 1$ to $n$}{
    \If{$d(p_i,C) > 2\min \{ \alpha r(p_i),  \Delta \}$}
    {
        $C := C \cup \{p_i\}$
    }
}
\Return $C$ 
\end{algorithm}
\DecMargin{1.5em}

\begin{lemma}\label{bound number centers}
Let  $k\in \NN$, $\alpha \in \RR_{> 0}$, $\Delta \in \RR_{>0}$ and $r: P \rightarrow \RR$ s.t. $r(p) \geq r_k(p)$ for all $p \in P$.
If there exists a feasible solution for $\alpha$-fair $k$-center with cost at most $\Delta$, then FairCenter$(k,\alpha,\Delta,r)$ returns a set $C$ of size at most $k$.
\end{lemma}

\begin{proof}
Let $(P,d)$ be the input instance and assume that there exists
a feasible solution $C^*$ for  the $\alpha$-fair $k$-center problem with cost at most $\Delta$ and let $C = \{c_1,\dots, c_\ell\}$ be the set returned by FairCenter$(k,\alpha,\Delta, r)$.
In order to prove the lemma, we will construct a set of $\ell$ disjoint balls (one for each center) such that every ball contains at least one center from $C^*$. Since $|C^*|\le k$ this will imply the result.

Assume that the centers $c_1,\dots,c_\ell$ are ordered increasingly by the time at which they were included in $C$. Note that, since $\alpha > 0$ and $P$ is sorted non-decreasingly with respect to $r$, for $1\leq i \leq j \leq \ell$, $\alpha r(c_i) \leq \alpha r(c_j)$ holds. Now define $r'(c_i):= \min\{\alpha r(c_i), \Delta \}$.
Since $c_1,\dots,c_\ell$ were included in $C$ by Line 4 of Algorithm  \ref{alg:FairCenter}, they satisfy 
$$
d(c_j,c_i) > 2 r'(c_j)
\geq  r'(c_i) +   r'(c_j)
$$
for $1 \leq i < j \leq \ell$,
which implies that
$$
B(c_i,r'(c_i)) \cap B(c_j, r'(c_j)) = \emptyset.
$$
We now observe that $C^*$ must have at least one center in each of these balls. Indeed, assume that $B(c_i, r'(c_i)) \cap C^*$ is empty for some $c_i \in C$. Then $d(c_i,C^*) > r'(c_i) = \min\{\alpha r(c_i),\Delta\} \geq \min\{\alpha r_k(c_i),\Delta\}$. Thus, $c_i$ either violates the fairness constraint of $C^*$ or the solution $C^*$ has cost more than $\Delta$, which is a contradiction.
Since the balls are disjoint and each ball contains at least one point of $C^*$, we conclude that $|C|\le k$.
\end{proof}

\begin{lemma}\label{bicriteria approximation}
Let  $k\in \NN,\alpha \in \RR_{> 0}$ and $\Delta\in\RR_{>0}$.
If there exists a feasible solution for $\alpha$-fair $k$-center with cost at most $\Delta$, then {FairCenter}$(k,\alpha,\Delta,r_k)$ returns a $(2\alpha,k)$-fair set of at most $k$ points with cost at most $2\Delta$.
\end{lemma}

\begin{proof}
Let $C$ be the set returned by FairCenter$(k,\alpha,\Delta,r_k)$. By Lemma \ref{bound number centers} we know that $|C| \leq k$. The if-clause in Line 4 ensures that $d(p,C) \leq 2  \min \{\alpha r_k(p), \Delta\} \leq 2 \Delta$ for all $p \in P$, which concludes the proof. 
\end{proof}

\IncMargin{1.5em}
\begin{algorithm}
\DontPrintSemicolon
\caption{$(2,2)$-ApproxFairCenter\label{alg:(2,2)-approx}}
\KwIn{Number of centers $k \in \NN$, fairness parameter $\alpha \in \RR_{> 0}$}
\KwOut{Subset of $P$ of size at most $k$, or failure message}
Let $D$ be set of pairwise distances in $P$

Compute $r_k(p)$ for $p\in P$ 

Do a binary search over $D$ to find smallest $\Delta$ s.t.
$|FairCenter(k,\alpha,\Delta,r_k)| \leq k$ 

\If{$\Delta$ does not exist}{
\Return \emph {no feasible solution exists}
}
\Else{
\Return $FairCenter(k,\alpha,\Delta, r_k)$
}

\end{algorithm}
\DecMargin{1.5em}

\begin{theorem}\label{thm: result discrete}

Let $k\in \NN$ and $\alpha \in \RR_{> 0}$.
If there exists a feasible solution for  $\alpha$-fair $k$-center, Algorithm \ref{alg:(2,2)-approx} computes a $(2,2)$-approximation for $\alpha$-fair $k$-center in time $O(n^2 + kn\log n)$ .

\end{theorem}
\begin{proof}
Observe that the cost of an optimal solution for $\alpha$-fair $k$-center is defined by the distance of two points in $P$. 
Let $\Delta^*$ be the cost of an optimal solution, then by Lemma \ref {bicriteria approximation} the set $C$ returned by FairCenter$(k,\alpha,\Delta^*,r_k)$ is a $(2,2)$-approximation for $\alpha$-fair $k$-center (note that we can include arbitrary additional points in $C$ until $|C|=k$, if $|C| < k$). 
To find the distance $\Delta^*$ that is passed to Algorithm~\ref{alg:FairCenter}, we first compute in $O(n^2)$ time for all points $p,q \in P$ their distances $d(p,q)$. For each input point we can then compute its $k$-fair radii $r_k(p)$ in $O(n)$ time (e.g. Theorem 9.3 in \cite{cormen2022introduction}).  
Additionally we sort $P$ in non-decreasing order with respect to the $k$-fair radii, which takes $O(n \log n)$ time.
We apply a binary search on the pairwise distances to find the optimal $\Delta^*$ and keep track of the subset of $D$ that can contain the optimal solution.
In each recursion step the median can be computed in linear time (e.g. Theorem 9.3 in \cite{cormen2022introduction}) without sorting the values a priori.
We invoke $O(\log n)$ queries to Algorithm \ref{alg:FairCenter} by doing so, where we increase the cost value if $|C| > k$ and decrease otherwise. 
If the binary search terminates without finding a cost value such that $|C| \leq k$ then no feasible solution exists. A call of Algorithm \ref{alg:FairCenter} takes $O(nk)$ time, since we computed the $k$-fair radii and sorted $P$ beforehand. The running time of the binary search on $1\leq t \leq n^2$ elements is $T(t) = O(t) + O(nk) + T(t/2) \in O(n^2 + kn \log n)$. This results in a total running time of $O(n^2 + k n \log n)$, which concludes the proof.
\end{proof}

Notice that all results in this section are under the assumption that a feasible solution exists.
We now remark what happens if this is not the case.
From Lemma \ref{alg:FairCenter} we can derive that if FairCenter$(k,\alpha,\Delta,r_k)$ returns a set of size larger than $k$ then there does not exist a feasible solution for $\alpha$-fair $k$-center clustering with cost at most $\Delta$. 
On the other hand, it may be the case that there does not exist a feasible solution but a set $C$ of size at most $k$ is returned. 
Since for all $p\in P$, $d(p,C) \leq 2 \min\{\alpha r_k(p),\Delta\}$, $C$ is then a feasible solution for $2 \alpha$-fair $k$-center with cost at most $2 \Delta$.

%% file: subquadratic_algorithm.tex
\section{FAST $(10,2+\varepsilon)$-APPROXIMATION}

Next, we give a subquadratic time $(10,2+\varepsilon)$-approximation algorithm. 
The running time bottleneck of the procedure discussed in Theorem \ref{thm: result discrete} is the computation of the candidate cost values and the fair radii, both of which take $\Theta(n^2)$ time. 
To overcome this, we show in Section \ref{subsec:valuelist} that a small set containing a good approximation of the optimal cost can be obtained using Gonzalez's algorithm \cite{gonzalez1985clustering} and we approximate the fair radii using uniform sampling in Section \ref{subsec:sampling}.

Since the randomized algorithm presented in this section has a running time in $\tilde{O}(kn)$, we get subquadratic time under the natural assumption that $k$ is asymptotically smaller then $n$.
Furthermore, in this section we assume the pairwise distances between points in $P$ to be distinct, which can be achieved by an arbitrarily small perturbation of the input points.

\subsection{COST VALUE CANDIDATES}\label{subsec:valuelist}
We start by discussing how to efficiently compute a small set of candidate values that contains a good approximation for the optimal cost of the given $\alpha$-fair $k$-center instance. To do so, we compute $k$ centers that are a constant approximation for the unconstrained $k$-center problem and show that the optimal cost of $\alpha$-fair $k$-center either relates to the cost of unconstrained $k$-center, or to the distance between two of our computed centers. We use the algorithm by Gonzalez for this procedure \cite{gonzalez1985clustering}. This algorithm finds a $2$-approximation for the optimal $k$-center cost by starting with an arbitrary point as first center and then greedily selecting the point as new center, which is furthest away from all previously chosen centers.

Algorithm \ref{algo: CostCandidates} describes our procedure to find a list of cost value candidates.

\IncMargin{1.5em}
\begin{algorithm}
\DontPrintSemicolon
\caption{CostCandidates\label{algo: CostCandidates}}
\KwIn{Number of centers $k \in \NN$, precision parameter $\varepsilon \in \RR_{>0}$}
\KwOut{Set $L \subset \RR$ }
Let $C_G$ be the set of points returned by Gonzalez's algorithm

$L := \{ \tfrac{1}{2}(1+\varepsilon)^j\max_{p\in P} d(p,C_G)\mid  j \in\{1,\dots,\lceil \log_{1+\varepsilon}(16)\rceil\} \}$

\For{$(g_1,g_2) \in C_G^2$ with $g_1 \neq g_2$}{
    $L := L \cup \{ \tfrac{1}{2}(1+\varepsilon)^j d(g_1,g_2)\mid j \in \{1,\dots,\lceil \log_{1+\varepsilon}(4) \rceil\} \} $
}
\Return $L$
\end{algorithm}
\DecMargin{1.5em}

\begin{lemma}
\label{lem: CostCandidates}
    Let $k\in\NN$ and $\alpha,\varepsilon \in\RR_{>0}$. Assume that there exists a solution for the $\alpha$-fair $k$-center problem and let $\Delta^*_\alpha$ be the cost of an optimal solution. CostCandidates$(k,\varepsilon)$ computes in $O(kn+k^2/\varepsilon)$ time a set $L$ of size $O(k^2/\varepsilon)$ such that there exists $\ell\in L$ with $\Delta^*_\alpha \leq \ell \leq(1+\varepsilon)\Delta^*_\alpha$.
\end{lemma}

\begin{proof}
    Let $\OPT$ be an optimal set of centers  for $k$-center (without fairness constraints) and $\Delta^* := \max_{p\in P} d(p,\OPT)$. Let $C_G$ be the set of centers found by Gonzalez's algorithm and $\Delta := \max_{p\in P} d(p,C_G)$. It holds that
    $ \Delta^* \leq \Delta \leq 2\Delta^*$
    and that $C_G$ can be computed in $O(kn)$ time as shown by Gonzalez \cite{gonzalez1985clustering}. 
      
    For each pair of distinct centers $g_1,g_2\in C_G$, let
    \[ \begin{aligned}
        L(g_1,g_2):=\{ &\tfrac{1}{2}(1+\varepsilon)^j d(g_1,g_2) \mid \\ &j \in  \{1, \dots, \lceil \log_{1+\varepsilon}(4) \rceil \}\}. 
    \end{aligned}\]
    Observe that the union of the intervals $[(1+\varepsilon)^{-1}\ell,\ell]$ over all $\ell\in L(g_1,g_2)$ completely covers the interval $[\tfrac{1}{2}d(g_1,g_2),2d(g_1,g_2)]$. 
    
    Furthermore, let
    \[ L_0 := \{ \tfrac{1}{2}(1+\varepsilon)^j\Delta \mid j \in \{1,\dots,\lceil \log_{1+\varepsilon}(16)\rceil\} \}. \]
    Similarly, the union of the intervals $[(1+\varepsilon)^{-1}\ell,\ell]$ over all $\ell\in L_0$ completely covers the interval $[\tfrac{1}{2}\Delta,8\Delta]$. Finally, let 
    \[ L := L_0\cup\bigcup_{g_1,g_2\in C_G,g_1\neq g_2}L(g_1,g_2)\] as it is constructed in Algorithm \ref{algo: CostCandidates}.
    
    We will show that $L$ has the desired properties. The bounds on running time and size are clear from the construction. To prove that there exists $\ell\in L$ satisfying $ \Delta^*_\alpha \leq \ell \leq(1+\varepsilon)\Delta^*_\alpha$, we consider two cases.

    First, assume that $$\Delta^*\leq\Delta^*_\alpha\leq 8\Delta^*.$$ Since $$\tfrac{1}{2}\Delta\leq\Delta^*\leq \Delta$$ it follows that $$\tfrac{1}{2}\Delta \leq \Delta^*_\alpha\leq 8\Delta.$$ Then there exists $j \in \{1,\dots, \lceil \log_{1+\varepsilon}(16)\rceil\}$ such that $$\tfrac{1}{2}(1+\varepsilon)^{j-1}\Delta\leq \Delta^*_\alpha\leq \tfrac{1}{2}(1+\varepsilon)^{j}\Delta,$$ thus, $$\Delta_\alpha^*\leq \tfrac{1}{2}(1+\varepsilon)^{j}\Delta \leq (1+\varepsilon)\Delta_\alpha^*.$$ 
    Since $\tfrac{1}{2}(1+\varepsilon)^j\Delta\in L$, this proves the claim.

    Second, assume that $\Delta^*_\alpha>8\Delta^*$. Let $\OPT_\alpha$ be an optimal solution for $\alpha$-fair $k$-center and consider $p\in P, c_\alpha\in \OPT_\alpha$ that realize $\Delta^*_\alpha$, i.e., $d(p,c_\alpha)=\Delta^*_\alpha$. Furthermore, let $g(p)$ and $g(c_\alpha)$ be the centers in $C_G$ closest to $p$ and $c_\alpha$, respectively. We have $d(p,g(p))\leq 2\Delta^*$ and $d(c_\alpha,g(c_\alpha))\leq 2\Delta^*$, so by the triangle inequality
    \begin{align*}
        d(g(p),g(c_\alpha))&\leq d(g(p),p)+d(p,c_\alpha)+d(c_\alpha,g(c_\alpha))\\
        &\leq \Delta^*_\alpha +4\Delta^*\\ 
        &\leq 2\Delta^*_\alpha   
    \end{align*}
    where the last inequality follows from $\Delta^*_\alpha>8\Delta^*$. On the other hand,  
    \begin{align*}
        d(g(p),g(c_\alpha))&\geq d(p,c_\alpha)-d(p,g(p))-d(c_\alpha,g(c_\alpha)),\\
        &\geq \Delta^*_\alpha-4\Delta^*,\\
        &\geq \tfrac{1}{2}\Delta^*_\alpha,
    \end{align*}
    where the last inequality again follows from $\Delta^*_\alpha>8\Delta^*$. Hence, 
    \[ \tfrac{1}{2}d(g(p),g(c_\alpha))\leq \Delta^*_\alpha \leq 2d(g(p),g(c_\alpha)). \]
    Then there exists $j\in \{1, \dots ,\lceil \log_{1+\varepsilon}(4) \rceil\}$ such that $$\tfrac{1}{2}(1+\varepsilon)^{j-1} d(g(p),g(c_\alpha))\leq \Delta^*_\alpha\leq \tfrac{1}{2}(1+\varepsilon)^{j}d(g(p),g(c_\alpha)),$$ thus, $$\Delta^*_\alpha\leq \tfrac{1}{2}(1+\varepsilon)^{j}d(g(p),g(c_\alpha)) \leq (1+ \varepsilon)\Delta_\alpha^*.$$ Since $$\tfrac{1}{2}(1+\varepsilon)^jd(g(p),g(c_\alpha))\in L(g(p),g(c_\alpha))\subseteq L,$$ this proves the claim in the second case, which concludes the proof.
\end{proof}

\subsection{APPROXIMATING FAIR RADII} \label{subsec:sampling}
Next we discuss how to efficiently compute good approximations for the $k$-fair radii of the points in $P$.
The procedure first samples a small subset of the input $S$ such that with constant probability it holds that for every point $p \in P$ there exists a point $q \in S$  with $d(p,q) \in (r_{3k}(p),r_k(p)]$. The pseudocode of this subprocedure is stated in Algorithm \ref{algo: fairsampling}. Note that a similiar result was proven by Czumaj and Sohler in the context of approximating the cost of $k$-nearest neighbor graphs (See lemma 3.1 in \cite{czumaj2024sublinear}). 

\IncMargin{1.5em}
\begin{algorithm}
\DontPrintSemicolon
\caption{FairSampling\label{algo: fairsampling}}
\KwIn{Number of centers $k \in \NN$ with $k\leq n/6$, failure probability $\delta \in (0,1)$}
\KwOut{Function $r': P \rightarrow \RR$}

$s:=  36  k \lceil    \ln(2 n / \delta) \rceil$\\
$t:= 27 \lceil   \ln(2 n / \delta) \rceil$\\
Let $S$ be a subset of size $s$ drawn uniformly at random with replacement from $P$ 

\For{$p \in P$}{
    Let $x$ be $t$-nearest neighbor of $p$ in $S$

    $r'(p) := d(p,x)$
}

\Return $r'$
\end{algorithm}
\DecMargin{1.5em}

\begin{lemma}
\label{lem: fairsampling}
    Given $k \in \NN$ such that $k\leq n/6$ and $\delta \in (0,1)$. Algorithm \ref{algo: fairsampling} returns in time $O(nk\log(n/\delta))$ with probability at least $1-\delta$
    for all $p\in P$ a value $r'(p)\in (r_{3k}(p),r_{k}(p)]$.
\end{lemma}

\begin{proof}
Let $S$ be a sample taken uniformly at random (with replacement) from $P$ of size $s =   36  k  \lceil  \ln(2 n / \delta)\rceil$ as in Algorithm \ref{algo: fairsampling}. We start by showing that $r'(p)\in (r_{3k}(p),r_{k}(p)]$ for all $p\in P$ with probability at least $1-\delta$. For $p \in P$ and $q \in S$ let $X^k_{pq}$ be the indicator random variable that is 1 if the sample $q$ is in $B(p,r_k(p))$ and 0 otherwise. Furthermore, let $X^{k}_p = \sum_{q \in S} X^k_{pq}$ be the random variable indicating the number of points of $S$ in $B(p,r_k(p))$. We define $X^{3k}_{pq}$ and $X^{3k}_p$ similarly by replacing $r_k(p)$ by $r_{3k}(p)$.

Observe that $r'(p)\in (r_{3k}(p),r_{k}(p)]$ if and only if $X^{3k}_p<t$ and $X^k_p\geq t$, since $r'(p)$ is defined as the distance from $p$ to its $t$-nearest neighbor. Therefore, $r'(p)\in (r_{3k}(p),r_{k}(p)]$ does \emph{not} hold only if $X^{3k}_p\geq t$ or $X^k_p< t$. 

First, we give an upper bound for $Pr(X^{3k}_p\geq t)$. By definition, $|B(p,r_{3k}(p))| = \lceil\frac{n}{3k}\rceil$, so 
\[\frac{n}{3k}\leq |B(p,r_{3k}(p))|\leq \frac{n}{3k}+1\leq \frac{n}{2k},\]
since $k\leq n/6$. Each $q \in S$ is sampled uniformly at random from $P$, so $\frac{1}{3k}\leq Pr(X^{3k}_{pq} = 1) \leq \frac{1}{2k}$. By linearity of expectation, 
$E[X^{3k}_p] = sPr(X^{3k}_{pq}=1)$, so $ 12\lceil \ln(2 n / \delta)\rceil \leq E[X^{3k}_p]\leq 18 \lceil \ln(2 n / \delta) \rceil$. Then
\begin{align*}
    Pr(X^{3k}_p\geq t)&=Pr(X^{3k}_p\geq 27 \lceil   \ln(2 n / \delta)\rceil),\\
    &\leq Pr(X^{3k}_p\geq \tfrac{3}{2}E[X^{3k}_p]),    
\end{align*}
from which it follows by a Chernoff bound (e.g. Theorem 4.4 
in~\cite{mitzenmacher2017probability}) that
\[ Pr(X^{3k}_p\geq t)\leq e^{-E[X^{3k}_p]/12}.\]
Since $E[X^{3k}_p]\geq 12\lceil \ln(2 n / \delta)\rceil$, we conclude that  
\[ Pr(X^{3k}_p\geq t)\leq e^{-\lceil   \ln(2 n / \delta)\rceil)}\leq \frac{\delta}{2n}.\]

Second, we give an upper bound for $Pr(X^k_p< t)$. By a similar reasoning as before, we see that $\frac{1}{k}\leq Pr(X^{k}_{pq} = 1) \leq \frac{3}{2k}$ and $ 36\lceil \ln(2 n / \delta)\rceil \leq E[X^{k}_p]\leq 54 \lceil \ln(2 n / \delta) \rceil$. Then
\begin{align*}
    Pr(X^{k}_p< t)&=Pr(X^{k}_p<27 \lceil   \ln(2 n / \delta)\rceil),\\
    &\leq Pr(X^{k}_p< \tfrac{3}{4}E[X^{k}_p]),    
\end{align*}
so it follows by a Chernoff bound (e.g. Theorem 4.5 
in~\cite{mitzenmacher2017probability}) that 
\[ Pr(X^{k}_p< t)\leq e^{-E[X^{k}_p]/32}.\]
Since $E[X^{k}_p]\geq 36\lceil \ln(2 n / \delta)\rceil$, we conclude that  
\[ Pr(X^{k}_p< t)\leq e^{-\lceil   \ln(2 n / \delta)\rceil)}\leq \frac{\delta}{2n}.\]

Combining both upper bounds, we see that $Pr(r'(p)\not\in (r_{3k}(p),r_{k}(p)])\leq \delta/n$. It follows from a union bound over all $p\in P$ that the probability that $r'(p) \notin (r_{3k}(p),r_k(p)]$ for at least one $p \in P$ is at most $\delta$, which proves the desired claim.

As for the running time, drawing the sample set $S$ takes $O(k \log (n/\delta))$ time. Finding the $t$-nearest neighbor of $p \in P$ in $S$ can be done in $O(k\log(n/\delta))$ time (e.g. Theorem 9.3 in \cite{cormen2022introduction}), yielding a total running time of $O(nk\log(n/\delta))$.
\end{proof}

Notice that the radii computed by Algorithm~\ref{algo: fairsampling} are not necessarily constant factor approximations of the fairness radii.
We can however use them to compute $3$-approximations of the actual $k$-fair radii as described in Algorithm \ref{alg:ApproxFairRadii}.
Algorithm \ref{alg:ApproxFairRadii} considers the input points in non-decreasing order with respect to the corresponding values $r'$ computed by Algorithm \ref{algo: fairsampling}.
Assuming Algorithm \ref{algo: fairsampling} to be successful, we argue that the following yields approximate fairness radii efficiently:
The algorithm maintains a set $C$ of
points for which we compute the fairness radii exactly. Initially $C$ is empty. 
Whenever a new point $p$ is considered, we check whether $B(p,r'(p))$ intersects any of the balls $B(q,r'(q))$ defined by the points $q\in C$. If this is not the case, we add $p$ to $C$. Otherwise, we set the radius value of $p$ to $d(p,q) + r_k(q)$. 
Since the balls defined by the $r'$ values contain sufficiently many points and are disjoint, we can argue that the  case of computing the $k$-fair radii explicitly does not happen too often. 

\IncMargin{1.5em}
\begin{algorithm}
\DontPrintSemicolon
\caption{ApproxFairRadii\label{alg:ApproxFairRadii}}
\KwIn{Number of centers $k \in \NN$ with $k\leq n/6$, failure probability $\delta \in (0,1)$}
\KwOut{Function $\tilde{r}:P \rightarrow \RR$, or a failure message} 
Let $r'$ be the function returned by FairSampling$(k,\delta)$

Organize $P$ in a min-heap $H$ with respect to $r'$

$p := \texttt{extractMin}(H)$

Compute $r_k(p)$ exactly and set $\tilde{r}(p):= r_k(p)$

$C := \{p\}$

\While{$H \neq \emptyset$}{
    $p := \texttt{extractMin}(H)$
    
    \If{$\exists q \in C ~: ~ r'(p) + r'(q) \geq d(p,q)$\label{line: case i}   }{
        $\tilde{r}(p) := d(p,q) + \tilde{r}(q)$  
    }
    \Else{
    \label{line: case ii}
    Compute $r_k(p)$ exactly and set $\tilde{r}(p):= r_k(p)$
    
    $C :=C \cup \{p\}$

    \If{$|C| > 3k$}
    {\label{line: check fair ball to small}\Return \emph{Fail}}
    }
}
\Return $\tilde{r}$
\end{algorithm}
\DecMargin{1.5em}

\begin{lemma} \label{lem: approx Fair Radii}
Given $k \in \NN$ with $k \leq n/6$ and $\delta \in (0,1)$. Then Algorithm \ref{alg:ApproxFairRadii} computes in $O(kn\log(n / \delta))$ time with probability at least $1 -\delta$  for all $p\in P$ a value $\tilde{r}(p)$ s.t. $r_k(p) \leq \tilde{r}(p)\leq 5 r_k(p)$.
\end{lemma}

\begin{proof}

    By Lemma \ref{lem: fairsampling}, $r'(p) \in (r_{3k}(p), r_k(p)]$ for all $p\in P$ with probability at least $1-\delta$, so assume that this is the case for the proof of correctness. The definition of $k$-fair radius implies the following claim.
    
    \begin{claim}\label{claim: fair radii def.}
       For all $p,q \in P$, $r_k(p) \leq d(p,q) + r_k(q)$. 
    \end{claim}

    Consider an arbitrary iteration of the algorithm and let $p$ be the point extracted from the min-heap. Note that for all $q\in C$, $\tilde{r}(q) = r_k(q)$. Let $(i)$ be the case that there exists a point $q \in C$ such that $d(p,q) \leq r'(p) + r'(q)$.  This corresponds to Line \ref{line: case i} of the algorithm. Note that $r'(q) \leq r'(p)$ for all $q \in C$, since these points were considered before $p$. Then  
    \begin{align*}
        \Tilde{r} (p) &=  d(p,q) + \tilde{r}(q)\\
                          &=    d(p,q) + r_k(q),\\
                          &\leq 2d(p,q) + r_k(p),\\
                          &\leq 2(r'(p) + r'(q)) + r_k(p),\\
                          &\leq 4 r'(p) + r_k(p),\\
                          &\leq 5 r_k(p),
         \end{align*}
     where the first inequality holds by Claim \ref{claim: fair radii def.}.
     Let $(ii)$ be the other case when there is no such point in $C$. This corresponds to Line \ref{line: case ii} of the algorithm. Here, we trivially get $\tilde{r}(p) = r_k(p)$.
    
    Let us consider the running time. By Lemma \ref{lem: fairsampling}, 
    the procedure FairSampling$(k,\delta)$ takes  $O(nk\log( n /\delta))$ time. Setting up the min-heap takes $O(n)$ time and extracting a point takes $O(\log n)$ time.
    Consider an arbitrary iteration of the algorithm.
    Observe that the running time of Case $(i)$ is $O(|C|)$, since we search $C$ for an intersection. Case $(ii)$ takes $O(n\log(n))$ time, since we compute the $k$-fair radius exactly. 
    Note that the number of iterations where Case $(ii)$ happens corresponds directly to the size of $C$. It therefore remains to bound the size of $C$.
    
    Assume that FairSampling$(k,\delta)$ is successful, i.e., that $r'(p) \in (r_{3k}(p), r_k(p)]$ for all $p \in P$. Then the ball $B(p,r'(p))$ contains at least $n/3k$ points for all $p\in P$, which implies that for all $p\in P$ there are at least $n/3k$ points $q$ in $P$ such that $r'(p) + r'(q) \geq d(p,q)$.
    Since the Balls $B(q,r'(q))$ of points $q\in C$ are pairwise disjoint by definition of $C$, we get that
    whenever a point $p$ is included to $C$ through Case $(ii)$ there exist at least $n/3k$ points for which Case $(i)$ happens. This implies that Case $(ii)$ can only happen up to $3k$ times, so $|C| \leq 3k$. In total, we therefore incur a running time of $O(nk\log(n / \delta))$.
    
    Now consider the case that the procedure FairSampling failed, i.e., there exists $p \in P$ such that $r'(p)\notin (r_{3k}(p),r_k(p)]$. If $r'(p) \leq r_{3k}(p)$, we cannot argue that Case $(ii)$ happens at most $3k$ times during the algorithm. It is therefore checked in Line \ref{line: check fair ball to small} if $|C| > 3k$, which implies that FairSampling failed and the algorithm terminates after only executing $O(k)$ iterations of Case $(i)$. The running time is therefore also $O(nk\log(n / \delta))$ for the failure case.    
\end{proof}

\subsection{SUBQUADRATIC TIME APPROXIMATION ALGORITHM}
\label{subsec:subquadalgo}

We can now state and prove the main result of this work.

\IncMargin{1.5em}
\begin{algorithm}
\DontPrintSemicolon
\caption{$(10,2+\varepsilon)$-ApproxFairCenter\label{alg:(10,2+eps)-approx}}
\KwIn{Number of centers $k \in \NN$, fairness parameter $\alpha \in \RR_{> 0}$, failure probability $\delta \in (0,1)$}
\KwOut{Subset of $P$ of size at most $k$, or failure message}

\If{$k>n/6$ \textbf{or} $k^2/\varepsilon > n^2 \log n$} {\Return $(2,2)$-ApproxFairCenter$(P,k,\alpha)$}

$\tilde{r} \gets \text{ApproxFairRadii}(k,\delta)$ 

\If{$\tilde{r} = \text{Fail}$}{
    \Return\emph{Fail}
}

$L \gets \text{CostCandidate}(k,\varepsilon/2)$

Binary search over $L$ to find smallest $\Delta$ s.t.
$|FairCenter(k,\alpha,\Delta,\tilde{r})| \leq k$ 

\If{$\Delta$ does not exist}{
\Return \emph {no feasible solution exists}
}
\Else{
\Return $FairCenter(k,\alpha,\Delta,r_k)$
}

\end{algorithm}
\DecMargin{1.5em}

\begin{theorem}\label{thm: subquadratic time}
    Given parameters $k \in \NN$, $\alpha, \varepsilon \in \RR_{>0}$ and $\delta \in (0,1)$. 
    If there exists a feasible solution for $\alpha$-fair $k$-center Algorithm \ref{alg:(10,2+eps)-approx} computes with probability at least $1-\delta$ a $(10,2+\varepsilon)$-approximation for the $\alpha$-fair $k$-center problem.
    The running time of the algorithm is in $O(kn \log(n/\delta) +k^2/\varepsilon)$ time. 
\end{theorem}

\begin{proof}
Assume that $k\leq n/6$ and $k^2/\varepsilon \leq n^2 \log n$, otherwise the proof follows from Theorem \ref{thm: result discrete}.

We first execute ApproxFairRadii$(k,\delta)$. By  Lemma \ref{lem: approx Fair Radii}, ApproxFairRadii$(k,\delta)$ outputs values $\tilde{r}(p)$ such that $r_k(p) \leq \tilde{r}(p) \leq 5r_k(p)$ for all $p \in P$ with probability at least $1-\delta$.

Next we execute CostCandidate$(k,\varepsilon/2)$. By Lemma \ref{lem: CostCandidates}, the set $L$ computed by CostCandidates$(k,\varepsilon/2)$ contains $\Delta \in L$ with $cost(P,C^*) \leq \Delta \leq (1+\varepsilon/2)\cost(P,C^*)$, where $C^*$ is an optimal solution for $\alpha$-fair $k$-center. 

Let $C$ be the solution returned by FairCenter$(k,\alpha,\Delta',\tilde{r})$ for an arbitrary $\Delta'\ge \Delta$. By Lemma \ref{bound number centers}, we know that in this case $|C| \leq k$. Thus, the binary search stops with a value $\Delta' \le \Delta$.
The if-clause in Line $4$ of Algorithm FairCenter ensures that 
\begin{align*}
    d(p,C) &\leq 2 \min\{\alpha \tilde{r}(p),\Delta'\} \\
    &\leq 2 \min\{\alpha 5r_k(p),(1+\varepsilon/2) cost(P,C^*)\}
\end{align*}
 This implies that $C$ is a $(10,2+\varepsilon)$-approximation for $\alpha$-fair $k$-center.
 
 Finally, we analyze the running time of this procedure. By Lemma \ref{lem: fairsampling}, ApproxFairRadii$(k,\delta)$ takes $O(kn\log(n / \delta))$ time. By Lemma \ref{lem: CostCandidates}, CostCandidates$(k,\varepsilon/2)$ requires $O(kn + k^2/\varepsilon)$ time. Sorting $P$ takes $O(n\log n)$ time. 
 Analogous to Algorithm \ref{alg:(2,2)-approx} one can implement the binary search to have running time $O(k^2/\varepsilon + kn \log(k/\varepsilon))$. 
 Under our assumptions the procedure takes $O(kn\log(n / \delta)+ k^2/\varepsilon)$ time in total.
\end{proof}